\documentclass[
a4paper,
aps,
pra,
showpacs,
twocolumn,
superscriptaddress]
{revtex4-2} 

\usepackage[utf8]{inputenc}  
\usepackage[T1]{fontenc}     
\usepackage[american]{babel}  
\usepackage[svgnames,psnames]{xcolor}
\usepackage[colorlinks,
citecolor=FireBrick,
linkcolor=blue,
urlcolor=blue]{hyperref}
\usepackage{graphicx} 
\usepackage{amsmath,amssymb,amsthm,bm,mathtools,amsfonts,mathrsfs,bbm,dsfont} 
\usepackage{listings}
\usepackage{array}

\newtheorem{theorem}{Theorem}
\newtheorem{observation}[theorem]{Observation}

\newtheorem{conjecture}[theorem]{Conjecture}
\newcommand{\ketbra}[2]{\left| #1 \right \rangle \left \langle #2 \right |  }

\newcommand{\tr}{{\mathrm{tr}}}

\newcommand{\ex}[1]{\ensuremath{\langle{#1}\rangle}}

\newcommand{\mean}[1]{\ensuremath{\left\langle{#1}\right\rangle}}
\newcommand{\eins}{\mathbbm{1}}

\renewcommand{\vr}{\ensuremath{\varrho}}
\renewcommand{\vec}[1]{\ensuremath{\boldsymbol{#1}}}
\usepackage{braket}
\usepackage{graphicx, xcolor}
\newcommand{\forget}[1]{}
\newcolumntype{P}[1]{>{\centering\arraybackslash}p{#1}}
\usepackage{comment}

\begin{document}
\title{Geometry of two-body correlations in three-qubit states}

\author{Shravan Shravan}
\affiliation{Center for Quantum Information and Control, Department of Physics and Astronomy, University of New Mexico, Albuquerque, NM, USA}

\author{Simon Morelli}
\affiliation{BCAM - Basque Center for Applied Mathematics,
Mazarredo 14, E48009 Bilbao, Basque Country - Spain}

\author{Otfried G\"uhne}
\affiliation{Naturwissenschaftlich-Technische Fakult\"at, Universit\"at Siegen, Walter-Flex-Stra\ss e 3, 57068 Siegen, Germany}

\author{Satoya Imai}
\affiliation{Naturwissenschaftlich-Technische Fakult\"at, Universit\"at Siegen, Walter-Flex-Stra\ss e 3, 57068 Siegen, Germany}

\date{\today}

\begin{abstract}
We study restrictions of two-body correlations in three-qubit states, using three local-unitarily invariant coordinates based on the Bloch vector lengths of the marginal states. First, we find tight nonlinear bounds satisfied by all pure states 
and extend this result by including the three-body correlations. Second, we consider mixed states and conjecture a tight non-linear bound for all three-qubit states. 
Finally, within the created framework we give criteria to detect different types 
of multipartite entanglement as well as characterize the rank of the quantum 
state.
\end{abstract}
\maketitle

\section{Introduction}
Processing quantum information relies on a thorough control of quantum systems of limited size. The states attainable by quantum systems form a set, called the quantum state space. A firm understanding of the underlying geometry of the state space hence becomes fundamental for developing efficient tools in quantum information processing.
There are remarkable differences between classical and quantum state spaces. The latter have a richer and more convoluted structure than their classical counterparts, which makes them difficult to characterize~\cite{bengtsson2017geometry}.
A famous example is the state space of a qubit, the so-called Bloch sphere, where all states are represented by a three-dimensional unit ball, each quantum 
state corresponding uniquely to a point of the ball \cite{bloch1946nuclear,Fano1954,Fano1957}.
A surface (or interior) point of this sphere corresponds to a pure (or mixed) state, and the length of its radius describes the state's mixedness, 
a quantity that remains invariant under unitary rotations.
However, the state spaces of higher-dimensional and multipartite quantum systems are complicated high-dimensional objects. 
It is hence fruitful to characterize them using only a restricted and accessible set of parameters for several reasons.
First, this compact description allows to see regularities and structures that would otherwise be obfuscated by the exponentially growing number of parameters.
Second, such {a reduced set of parameters might be easier} to {access experimentally and} keep track of, when a full description of a high-dimensional quantum state becomes experimentally unviable or computationally intractable.
Third, a reduced set of parameters allows to visualize relevant properties of the quantum state space in lower dimensions.
Many efforts have been devoted to the geometrical characterization of {the quantum state space} of
single qutrits \cite{kurzynski2016three, eltschka2021shape},
two qubits \cite{horodecki1996information, gamel2016entangled,morelli2023correlation},
three qubits \cite{wyderka2020characterizing,imai2021bound},
many qubits \cite{hiesmayr2008simplex,toth2007optimal,toth2012multipartite}
or other properties like
separable balls \cite{zyczkowski1998volume},
steering ellipsoids \cite{jevtic2014quantum},
and the
Majorana representation \cite{martin2010multiqubit}.

In this paper, we study the geometry of the admissible two-body correlations of a three-qubit system and derive 
important properties of the global state from its two-body correlations.
This is connected to the quantum marginal problem, where properties of a global multipartite state can be inferred 
from local properties of the parties and correlations between a reduced number of parties
\cite{linden2002almost,
diosi2004three,
jones2005parts,
wurflinger2012nonlocal, 
sawicki2013pure,
walter2013entanglement,
miklin2016multiparticle,
wyderka2017almost,
yu2021complete,
navascues2021entanglement,
huber2022refuting}.
Let $\vr_{ABC}$ be a three-qubit state and let $\vr_{AB} = \tr_C(\vr_{ABC})$ be its two-qubit reduced state shared between the systems $A$ and $B$. Similarly, consider the two-qubit reduced states shared between the systems $B-C$, $\vr_{BC}$, and between the systems $C-A$, $\vr_{CA}$.
Then, the total two-body correlations of the reduced states based on the Pauli matrices $\sigma_{i}$ for $i=1,2,3$
are given by
\begin{subequations}
\begin{align}
    S^{AB}_2
    &= \sum_{i,j = 1}^3
    \braket{\sigma_i \otimes \sigma_j}_{\vr_{AB}}^2,
    \\
    S^{BC}_2
    &= \sum_{i,j = 1}^3
    \braket{\sigma_i \otimes \sigma_j}_{\vr_{BC}}^2,
    \\
    S^{CA}_2
    &= \sum_{i,j = 1}^3
    \braket{\sigma_i \otimes \sigma_j}_{\vr_{CA}}^2.
\end{align}
\end{subequations}
The quantities $S^{XY}_2$ for $X,Y = A, B, C$ are known to be invariant under local unitaries and can be understood as {part of the total Bloch vector length, encoding the correlations between the corresponding subsystems}.
{We aim to characterize the state space of a three-qubit system with the three coordinates $(x, y, z) = (S^{AB}_2, S^{BC}_2, S^{CA}_2)$, thus obtaining a three-dimensional representation of the in principle $63$-dimensional state space.}

Several works have analyzed quantum state spaces in similar directions.
One research line is to take an average over all subsystems and consider the so-called two-body sector length 
$S_2 \equiv S^{AB}_2 + S^{BC}_2 + S^{CA}_2$. Similarly, one can define general $k$-body sector lengths for general
multipartite systems \cite{aschauer2003local, wyderka2020characterizing, eltschka2020maximum, imai2021bound, 
cieslinski2023analysing}. For two qubits there exists a full description of the restriction of the one- and 
two-body sectors~\cite{Morelli2020dimensionally, wyderka2020characterizing, morelli2023correlation}.
Moreover, for three qubits the admissible sector lengths $S_1, S_2, S_3$ form a three-dimensional polytope whose boundaries
are explicitly known \cite{wyderka2020characterizing}.
Another research line is to develop practical methods for accessing $S^{XY}_2$ or more general sector lengths without state tomography.
Recently, randomized measurement schemes with Haar random local unitaries have been used to obtain the sector length directly \cite{
brydges2019probing,
elben2023randomized,
cieslinski2023analysing}.
This method does not require spatially-separated parties to align their measurement directions, thus enabling reference-frame-independent 
quantum information processing.

{The paper is structured as follows.}
In Sec.~\ref{sec:pure}, we consider three-qubit pure states and find tight non-linear bounds on the two-body correlations defined above, 
improving the previously existing condition $S_2 = 3$.
In Sec.~\ref{sec:mixed}, we discuss the analysis of mixed states in the same coordinates.
First, we extend our previous solution to mixed states and conjecture a tight non-linear bound for all three-qubit states.
Next, we present systematic methods to obtain linear polytopes for mixed separable states.
These results allow us to detect multipartite entanglement from marginal correlations that 
are invariant under local unitaries. In Sec.~\ref{sec:rank}, we discuss conditions imposed 
on the two-body correlations based on the rank of the quantum states. Finally, we conclude and mention 
further problems worth to be studied.

\section{Pure states}\label{sec:pure}
In this section, we consider pure three-qubit states.
We begin by noting that the condition $S_2 = 3$ holds for any pure three-qubit 
state $\vr_{ABC}$~\cite{wyderka2020characterizing}. Indeed, using 
the  Schmidt decomposition, it is easy to see that for any bipartition $X,\Bar{X}$ 
of the system, 
$\operatorname{tr(\vr(\vr_X \otimes \mathds{1}_{\Bar{X}}))} 
= \operatorname{tr(\vr(\mathds{1}_X \otimes \vr_{\Bar{X}}))}$; 
meaning that the purities $\operatorname{tr}(\vr_X^2)$ and 
$\operatorname{tr}(\vr^2_{\Bar{X}})$ are the same.
Summing over the above expression for all the three bipartitions 
gives $S_2 = 3$.
At the same time, this is not the strongest restriction on the set of points $(S^{AB}_2, S^{BC}_2, S^{CA}_2)$ compatible with a pure three-qubit state.
We now present the first result of this paper: 
\begin{observation}\label{ob:proof_pure}
For any three-qubit pure state the linear condition $S^{AB}_2 + S^{BC}_2 + S^{CA}_2 = 3$ holds. It further holds that 
\begin{subequations}
       \begin{align}
        \begin{split}
            \label{eq:pure_nonlin1}
            \sqrt{S_2^{AB}} + \sqrt{S_2^{BC}} - \sqrt{S_2^{CA}}
            &\leq \sqrt{3}, 
        \end{split} \\
        \begin{split}
            \label{eq:pure_nonlin2}
            \sqrt{S_2^{AB}} - \sqrt{S_2^{BC}} + \sqrt{S_2^{CA}}
            &\leq \sqrt{3}, 
        \end{split} \\
        \begin{split}
            \label{eq:pure_nonlin3}
            -\sqrt{S_2^{AB}} + \sqrt{S_2^{BC}} + \sqrt{S_2^{CA}}
            &\leq \sqrt{3}, 
        \end{split}
       \end{align}
    \end{subequations}
\end{observation}

The proof of this Observation is given in Appendix~\ref{ap:proof_pure} and the set of points attainable by pure three-qubit states is displayed in Fig.~\ref{fig:proof_pure}.

\begin{figure}[t]
    \centering
    \includegraphics[width=0.8\columnwidth]{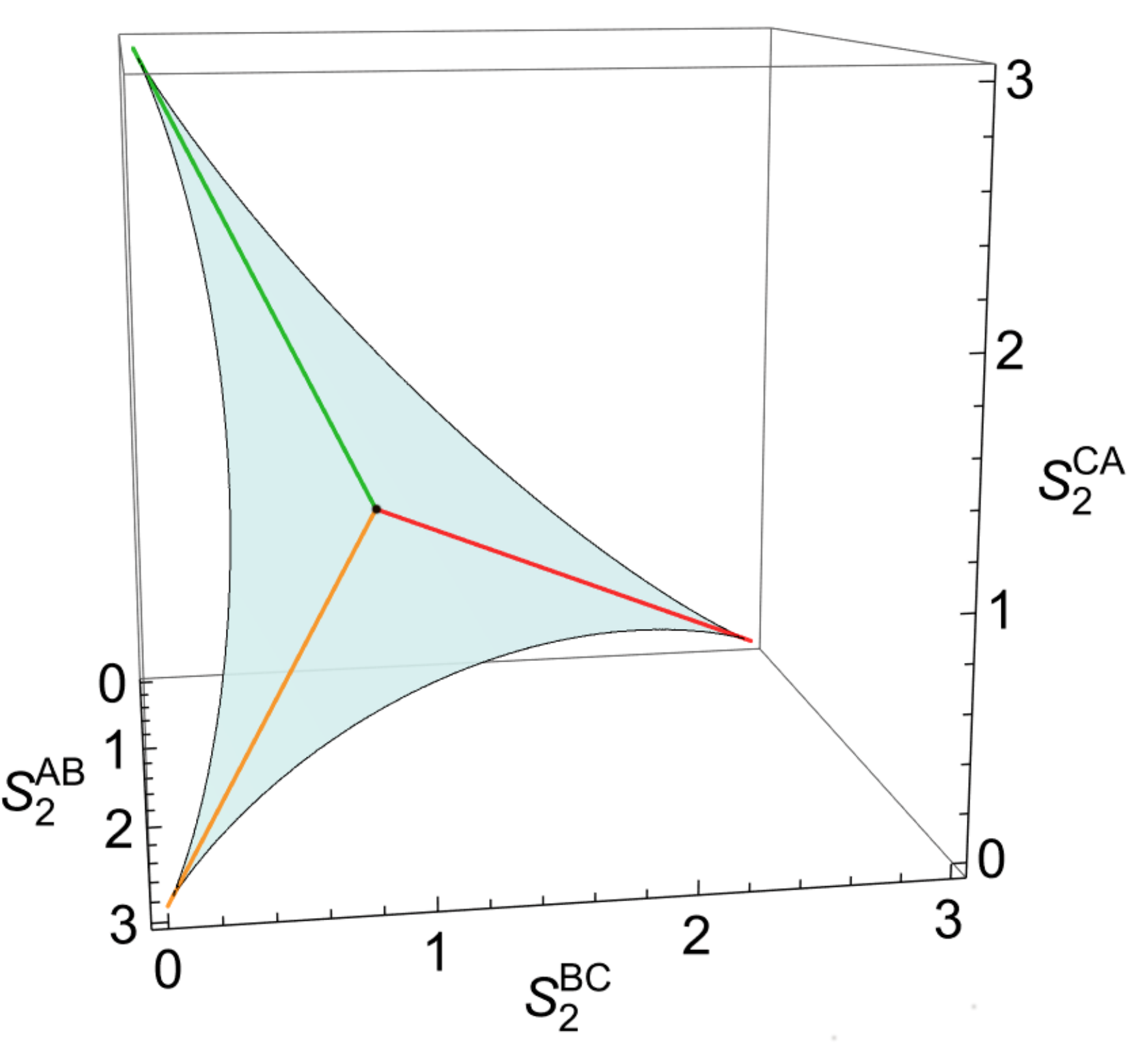}
    \caption{
    The nonlinear area of the pure three-qubit states in the coordinate space $(S^{AB}_2, S^{BC}_2, S^{CA}_2)$ based on Observation~\ref{ob:proof_pure}.
    The Orange, Red, and Green lines, respectively represent the states that are biseparable in the $AB|C$, $BC|A$, and $CA|B$ partitions.
    The black dot contains the fully separable states}
    \label{fig:proof_pure}
\end{figure}

There are several noteworthy facts about Observation~\ref{ob:proof_pure}. 
The non-linear inequalities in Eqs.~(\ref{eq:pure_nonlin1})-(\ref{eq:pure_nonlin3}) describe a two-dimensional region that is strictly smaller than the previously known region described by $S_2 = 3$, see Fig.~\ref{fig:proof_pure}. This implies that the amount of two-qubit correlations in a three-qubit state constrain each other in a non-linear fashion. In fact, the inequalities are tight, for every point in this region there exists a pure state that is mapped to this point. Each inequality in Eqs. ~ (\ref{eq:pure_nonlin1})-(\ref{eq:pure_nonlin3}) is correspondingly saturated by  the following state families belonging to the W class:
\begin{subequations}
    \begin{align}
     \begin{split}
         \label{eq:pure_1_state}
          \ket{\xi_1 (p)}&= p\ket{001}+f_p^+\ket{010} +f_p^-\ket{100},
     \end{split}\\
     \begin{split}
     \label{eq:pure_2_state}
           \ket{\xi_2 (p)} &= p\ket{010}+f_p^+\ket{100}
          +f_p^-\ket{001},      
     \end{split} \\
     \begin{split}
     \label{eq:pure_3_state}
           \ket{\xi_3 (p)} &= p\ket{100}+f_p^+\ket{001}
          +f_p^-\ket{010},
     \end{split} 
 \end{align}
\end{subequations}
where $0\le p\le1/\sqrt{2}$ and $f_p^{\pm} = (\sqrt{2- 3p^2}\pm p)/2$. Note that there exist mixed states that also satisfy Eqs.~ (\ref{eq:pure_nonlin1})-(\ref{eq:pure_nonlin3}) and $S_2 = 3$. One such example is the state $\vr = \frac{1}{4} \left(\ket{00} + \ket{11}\right)\left(\bra{00} + \bra{11}\right) \otimes \mathds{1}$, which has the coordinates $(3,0,0)$. 

If a state is separable with respect to a bipartition $A|BC$, {written as $\ket{\psi_{A|BC}}=\ket{\psi_{A}} \otimes\ket{\psi_{BC}}$}, then it follows that $S^{AB}_2 = S^{CA}_2 \leq 1$.
This can be shown using the condition $S^{AB}_2 + S^{BC}_2 + S^{CA}_2 = 3$ and the product formula $S_2^{XY}(\vr_{X}\otimes \vr_{Y}) = S_1^X(\vr_X)S_1^Y(\vr_Y)$ \cite{wyderka2020characterizing} for single-particle {states} $\vr_{X}$ and $\vr_{Y}$.
Similarly, one can find the cases for the other two bipartitions $B|CA$ and $C|AB$.
These separability conditions for $A|BC$, $B|CA$, and $C|AB$ are respectively represented by the Red, Green, and Orange lines in Fig.~\ref{fig:proof_pure}, where a state on each line corresponds, up to local unitaries, to one of the following biseparable states
\begin{equation}\label{eq:pure_bisep}
    \ket{\phi(\theta)}_{X|YZ}= \ket{0}_X \otimes (\cos{(\theta)}\ket{00}_{YZ} + \sin{(\theta)}\ket{11}_{YZ})
\end{equation}
for $\theta \in [0, \pi/2]$ and $X,Y = A,B,C$.
Since it has to obey all three conditions, any full-product state $\ket{\psi_{\text{fs}}} = \ket{a}\otimes\ket{b}\otimes\ket{c}$ is mapped to the center point $(1,1,1)$ of the shape.
Any pure state mapped outside of these lines is genuinely three-qubit entangled.

However, the converse is not true. There exist entangled states that are mapped to the biseparable lines.
If a three-qubit state $\vr_{ABC}$ is permutationally invariant under the exchange of two qubits $X, Y =A, B, C$, that is, $F_{XY} \vr_{ABC} F_{XY} = \vr_{ABC}$ with the flip operator between $X, Y$ qubits, then it holds that $S_2^{XZ} = S_2^{YZ}$, which can be immediately shown using $F (A \otimes B) F = B \otimes A$ for any operators $A, B$.
States that are invariant under the exchange of all three qubits such as the symmetric Greenberger–Horne–Zeilinger (GHZ) state $\ket{\text{GHZ}}=(\ket{000}+\ket{111})/\sqrt{2}$ and the W state $\ket{\text{W}}=(\ket{001}+\ket{010}+\ket{100})/\sqrt{3}$ are mapped to the point $(1,1,1)$.
This means that any pure state with some {exchange}-symmetry cannot be distinguished from a biseparable or fully separable state from its two-body correlations alone.
In the following, we will overcome this issue by adding the three-body correlations as another {parameter}.

To proceed, let us first define the three-body sector length as
\begin{equation}
    S_3^{ABC}
    = \sum_{i,j,k = 1}^3
    \braket{\sigma_i \otimes \sigma_j \otimes \sigma_k}_{\vr_{ABC}}^2.
\end{equation}
Again, this is invariant under local unitaries.
{We then} make the following observation:
\begin{observation}\label{ob:proof_threedim}
For any three-qubit pure state it holds that 
\begin{equation}\label{eq:pure_3d1}
    S_3^{ABC}\le 3+ \min(S_2^{AB},S_2^{BC},S_2^{CA}),
\end{equation}
and
\begin{subequations}
\label{eq:pure_3d}
\begin{align}
    \sqrt{S_2^{AB} + \Delta} + \sqrt{S_2^{BC}+ \Delta}
    - \sqrt{S_2^{CA}+ \Delta}
    &\le \sqrt{3},\\
    \sqrt{S_2^{AB} + \Delta} - \sqrt{S_2^{BC}+ \Delta}
    + \sqrt{S_2^{CA}+ \Delta}
    &\le \sqrt{3},\\
    -\sqrt{S_2^{AB} + \Delta} + \sqrt{S_2^{BC}+ \Delta}
    + \sqrt{S_2^{CA}+ \Delta}
    &\le \sqrt{3},
\end{align}
\end{subequations}
where $\Delta \equiv 3 - S_3^{ABC} \in [-1, 2]$.
\end{observation}

The proof of this Observation is also given in Appendix~\ref{ap:proof_pure} and the geometrical object defined by these inequalities is shown in Fig.~\ref{fig:proof_threedim}.

\begin{figure}
    \centering
    \includegraphics[width = \columnwidth]{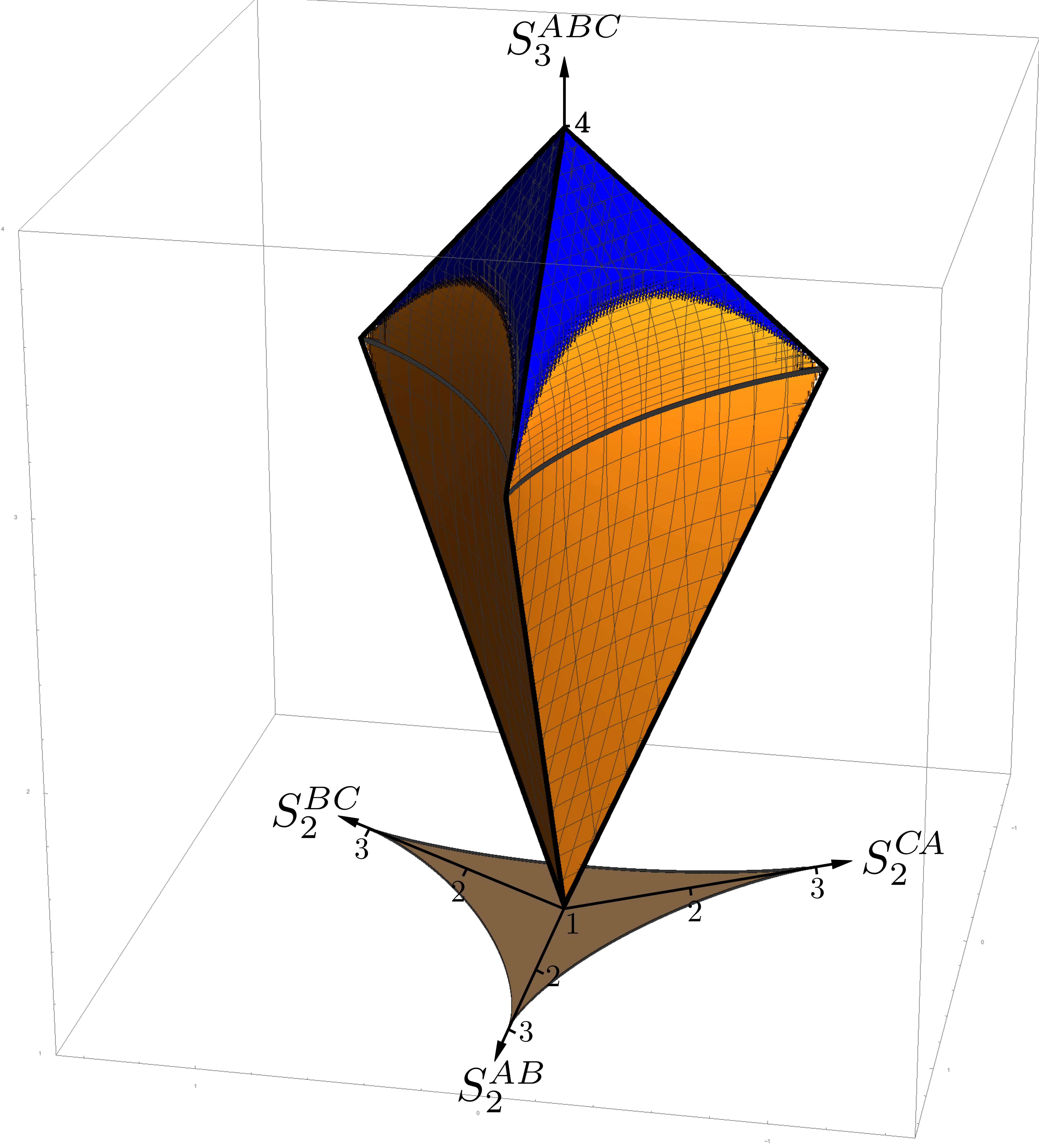}
    \caption{{The space of all two- and three-body correlations for pure three-qubit states based on Observation~\ref{ob:proof_threedim}. The figure shows the restrictions of the two-body correlations as in Fig.~\ref{fig:proof_pure} (gray shadow), extended by an additional coordinate, the three-body correlation $S_3^{ABC}$. The resulting object lies in a three-dimensional subspace of $\mathbb{R}^4$. For simplicity, we will use the four coordinates, but note that this orthonormal basis in $\mathbb{R}^4$ becomes a spanning set for the three-dimensional subspace. For example, the lower corner then has the coordinates $(1,1,1,1)$, and the upper corner has the coordinates $(1,1,1,4)$. The blue region is bounded by Eq.~\eqref{eq:pure_3d1} and the orange boundaries are given by Eq.~\eqref{eq:pure_3d}.}
    }
    \label{fig:proof_threedim}
\end{figure}

We have several remarks on Observation~\ref{ob:proof_threedim}.
First, similarly to the result in Observation~\ref{ob:proof_pure}, these non-linear inequalities are tight, they exactly describe the set of all points corresponding to pure states.
The boundary (colored orange in Fig.~\ref{fig:proof_threedim}) is attained by states of the form
\begin{equation}
    \label{eq:pure_3d2}
  \ket{\Xi(x,y)}= \sqrt{\frac{x-y}{2}}|001\rangle+\sqrt{\frac{1+y}{2}}|010\rangle+\sqrt{\frac{1-x}{2}}|100\rangle,
\end{equation}
where $-1 < y \leq x <  1$.
As shown in Fig.~\ref{fig:proof_threedim}, the intersection with horizontal planes of height $S_3^{ABC}$ is minimized for $S_3 = 4$ and $S_3^{ABC} = 1$ and maximized for $S_3^{ABC} = 3$. Thus the projection onto a horizontal plane (Observation~\ref{ob:proof_pure}) gives the same shape as the intersection with the plane at $S_3^{ABC} = 3$.
Observation~\ref{ob:proof_threedim} shows that there is a fundamental limitation in quantum mechanics regarding the allowable differences between two-body correlations for given a three-body correlation, which can be represented geometrically as the area of a hyperbolic triangle formed by the intersection of a hyperbolic triangular bipyramid with a plane.

Second, any pure separable state with respect to the bipartition $A|BC$ satisfies $S_3^{ABC} \leq 3$ {and} is mapped to the line connecting the point $(1,1,1,1)$ with $(3,0,0,3)$. To see this, note that the local states in $A$ and $BC$ have to be pure and therefore $S_1^{A}=1$ and $S_2^{BC}=3-S_1^{B}-S_1^{C}$. So it follows that $S_2^{BC}=S_3^{ABC}$, since for pure states it holds  that $S_3^{ABC}+S_1^{A}+S_1^{B}+S_1^{C}=4$~\cite{wyderka2020characterizing}. The same holds true for the other partitions.
From this, it immediately follows that all pure fully separable states satisfy $S_3^{ABC} = 1$ and are mapped to the single point $(1,1,1,1)$. This also follows from the fact that all three-qubit fully separable states satisfy $S_3 \leq 1$~\cite{wyderka2020characterizing, imai2021bound}.
Thus, by including the $S_3^{ABC}$ coordinate, we are able to clearly distinguish between fully separable, biseparable, and multipartite entangled states.
For example the GHZ state has now the coordinates $(1,1,1,4)$ and the W state has the coordinates $(1,1,1,11/3)$.

\section{Mixed states}\label{sec:mixed}
In this section, we extend our previous analysis to mixed states and characterize the set of two-body correlations $(S^{AB}_2, S^{BC}_2, S^{CA}_2)$ attainable by various classes of mixed three-qubit states $\vr_{ABC}$.
Let us begin by noting the convexity property of the two-body sector lengths:
$S_2(\vr_{ABC}) \leq \sum_i p_i S_2(\vr_i)$ for $\vr_{ABC} = \sum_i p_i \vr_i$ (this holds for all $k$-body sector lengths).
Thus, any three-qubit state obeys that $S_2(\vr_{ABC}) \leq 3$.
This means that no three-qubit mixed state can go outside the boundaries of pure states in the state space.

For any three-qubit state the inequality $S^{AB}_2 + S^{BC}_2 + S^{CA}_2 \le 3$ holds.
We further conjecture that the same three non-linear inequalities as for pure states given in Observation~\ref{ob:proof_pure} hold:
\begin{conjecture}\label{conj:mixed_total _nonlin}
For any three-qubit state it holds that
\begin{subequations}
       \begin{align}
        \begin{split}
            \label{Full_ss_b}
            \sqrt{S_2^{AB}} + \sqrt{S_2^{AC}} - \sqrt{S_2^{BC}}
            &\leq \sqrt{3}, 
        \end{split} \\
        \begin{split}
            \label{Full_ss_c}
            \sqrt{S_2^{AB}} - \sqrt{S_2^{AC}} + \sqrt{S_2^{BC}}
            &\leq \sqrt{3}, 
        \end{split} \\
        \begin{split}
            \label{Full_ss_d}
            -\sqrt{S_2^{AB}} + \sqrt{S_2^{AC}} + \sqrt{S_2^{BC}}
            &\leq \sqrt{3}.
        \end{split}
       \end{align}
    \end{subequations}
\end{conjecture}

The corresponding set $(S^{AB}_2, S^{BC}_2, S^{CA}_2)$ is displayed in Fig.~\ref{fig:conj_anymix}.
There is some evidence indicating the validity of this conjecture.
One expects states that form the boundary of the state space to be of low rank, whereas states increase in rank when coming closer to the origin $(0,0,0)$, see Sec.~\ref{sec:rank}.
But Conjecture~\ref{conj:mixed_total _nonlin} holds true for rank-two mixed states, for details see Appendix~\ref{ap:mixed_total _nonlin}.

\begin{figure}[t]
    \centering
    \includegraphics[width=0.8\columnwidth]{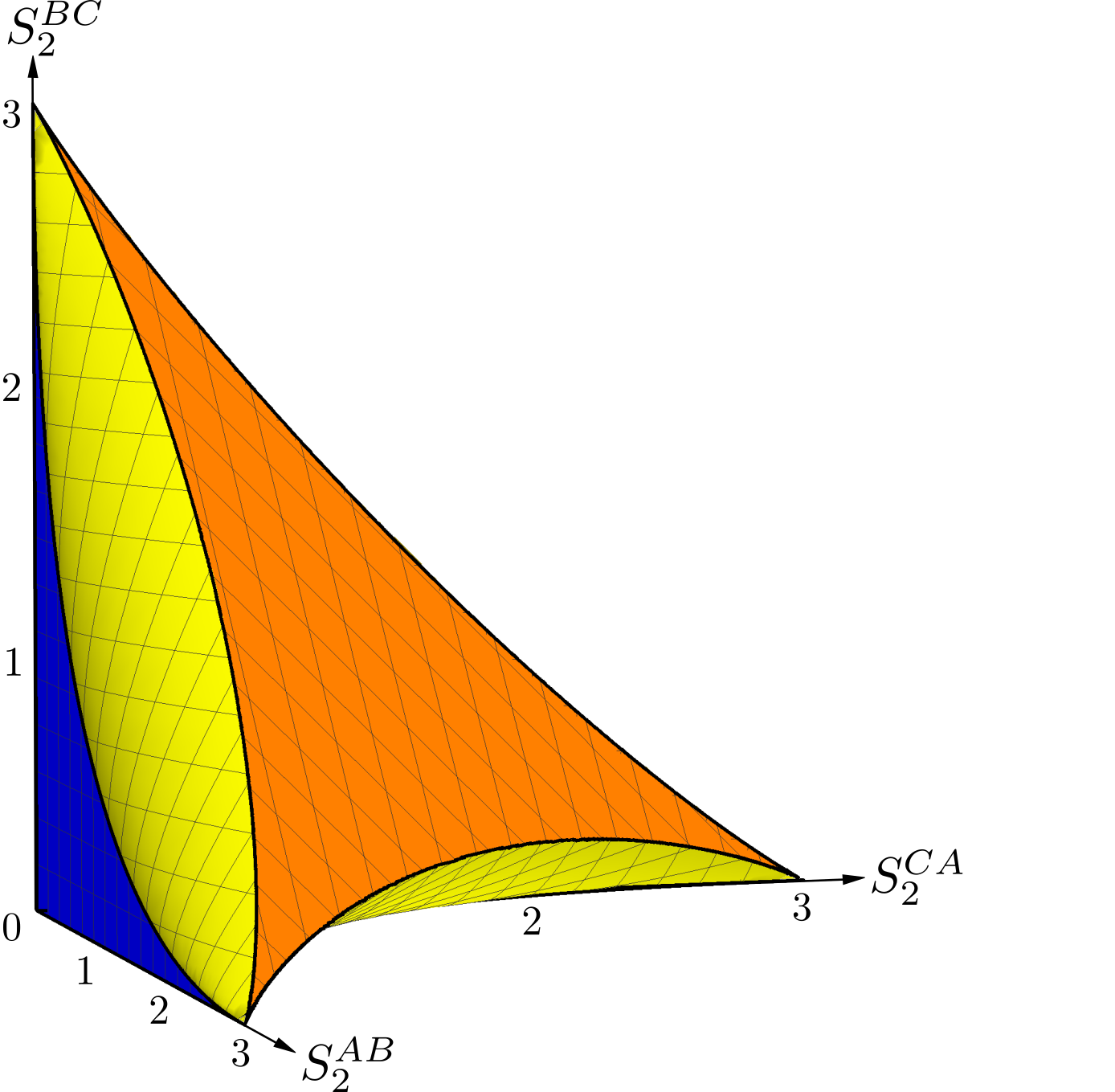}
    \caption{The conjectured state space of mixed states of three qubits under two-body sector lengths based on Conjecture \ref{conj:mixed_total _nonlin}. }
    \label{fig:conj_anymix}
\end{figure}

In case the inequalities~(\ref{Full_ss_b}, \ref{Full_ss_d}, \ref{Full_ss_d}) hold, they would be tight, as they are saturated by the states
\begin{subequations}
    \begin{align}
        \begin{split}
            \kappa_{1}(p,q)
            &= q \ket{\xi_1(p)}\!\bra{\xi_1(p)}
            + \Bar{q} \ket{\xi_1(\Bar{p})}\!\bra{\xi_1(\Bar{p})},
        \end{split} \\
        \begin{split}
            \kappa_{2}(p,q)
            &= q \ket{\xi_2(p)}\!\bra{\xi_2(p)}
            + \Bar{q} \ket{\xi_2(\Bar{p})}\!\bra{\xi_2(\Bar{p})},
        \end{split} \\
        \begin{split}
            \kappa_{3}(p,q)
            &= q \ket{\xi_3(p)}\!\bra{\xi_3(p)}
            + \Bar{q} \ket{\xi_3(\Bar{p})}\!\bra{\xi_3(\Bar{p})},
        \end{split} 
    \end{align}
\end{subequations}
where $\Bar{q} = 1-q$, $\Bar{p} = 1-p$ and $\ket{\xi_1(p)},\ket{\xi_2(p)},\ket{\xi_3(p)}$ are defined in Eqs.~(\ref{eq:pure_1_state}, \ref{eq:pure_2_state}, \ref{eq:pure_3_state}).
That is, together with $S_2 \le 3$ they would completely characterize the set of points $(S^{AB}_2, S^{BC}_2, S^{CA}_2)$ attainable by three-qubit states.

A mixed state is called fully separable if it can be written in the form
\begin{equation}
    \vr_{\text{fs}} =
    \sum_k p_k \vr_k^A \otimes \vr_k^B \otimes \vr_k^C,
\end{equation}
with $p_k \in [0,1]$ and $\sum_k p_k =1$.
For example, the maximally mixed state $\eins_8/8$ sits at the origin $(0,0,0)$.
Intuitively, a fully separable state should have small two-body correlations, and thus the set $(S^{AB}_2, S^{BC}_2, S^{CA}_2)$ attainable by fully separable states may be expected to sit close to the origin.
Indeed, we present the following inequalities:
\begin{observation}\label{ob:Fully_separable_mixed}
Any fully separable three-qubit state obeys
\begin{subequations}
\label{fs}
 \begin{align}
     \begin{split}
     \label{fs_b}
         S^{AB}_2 + S^{BC}_2 - S^{CA}_2 &\leq 1,
     \end{split} \\
    \begin{split}
     \label{fs_c}
           S^{AB}_2 - S^{BC}_2 + S^{CA}_2 &\leq 1,
     \end{split} \\
     \begin{split}
     \label{fs_d}
           -  S^{AB}_2 + S^{BC}_2 + S^{CA}_2 &\leq 1.
     \end{split}
 \end{align}
\end{subequations}
These criteria are optimal, the fully separable states fill the whole set defined by those inequalities.
\end{observation}

The proof is given in Appendix~\ref{ap:Fully_separable_mixed}.
The set defined formed by the fully separable states is displayed in Fig.~\ref{fig:Fully_separable_mixed}, the inequalities in Eqs.~(\ref{fs_b})-(\ref{fs_d}) are respectively represented by the Blue, Green, and Pink facets.
In the case of $S_2 = 3$, this set reduces to the point $(1,1,1)$, recovering the corresponding condition for pure states.
These inequalities are tight, they are saturated by the states
\begin{subequations}
 \label{eq:fs_s}
 \begin{align}
     \begin{split}
     \label{fs_s_b}
       \varrho_{++-}(p,\theta)
       &= p \ket{000}\!\bra{000} 
       + (1-p) \ket{01\theta}\!\bra{01\theta}, 
     \end{split} \\
    \begin{split}
     \label{fs_s_c}
       \varrho_{+-+}(p,\theta)
       &= p \ket{000}\!\bra{000} 
       + (1-p) \ket{0\theta1}\!\bra{0\theta1},
     \end{split} \\
     \begin{split}
     \label{fs_s_d}
      \varrho_{-++}( p,\theta)
      &= p \ket{000}\!\bra{000} 
      + (1-p) \ket{\theta 01}\!\bra{\theta 01},
     \end{split}
 \end{align}
\end{subequations}
where $\ket{\theta}
= \cos{( \theta )} \ket 0 + \sin{(\theta )} \ket 1 $. 

\begin{figure}
  \centering
  \includegraphics[width=0.8\columnwidth]{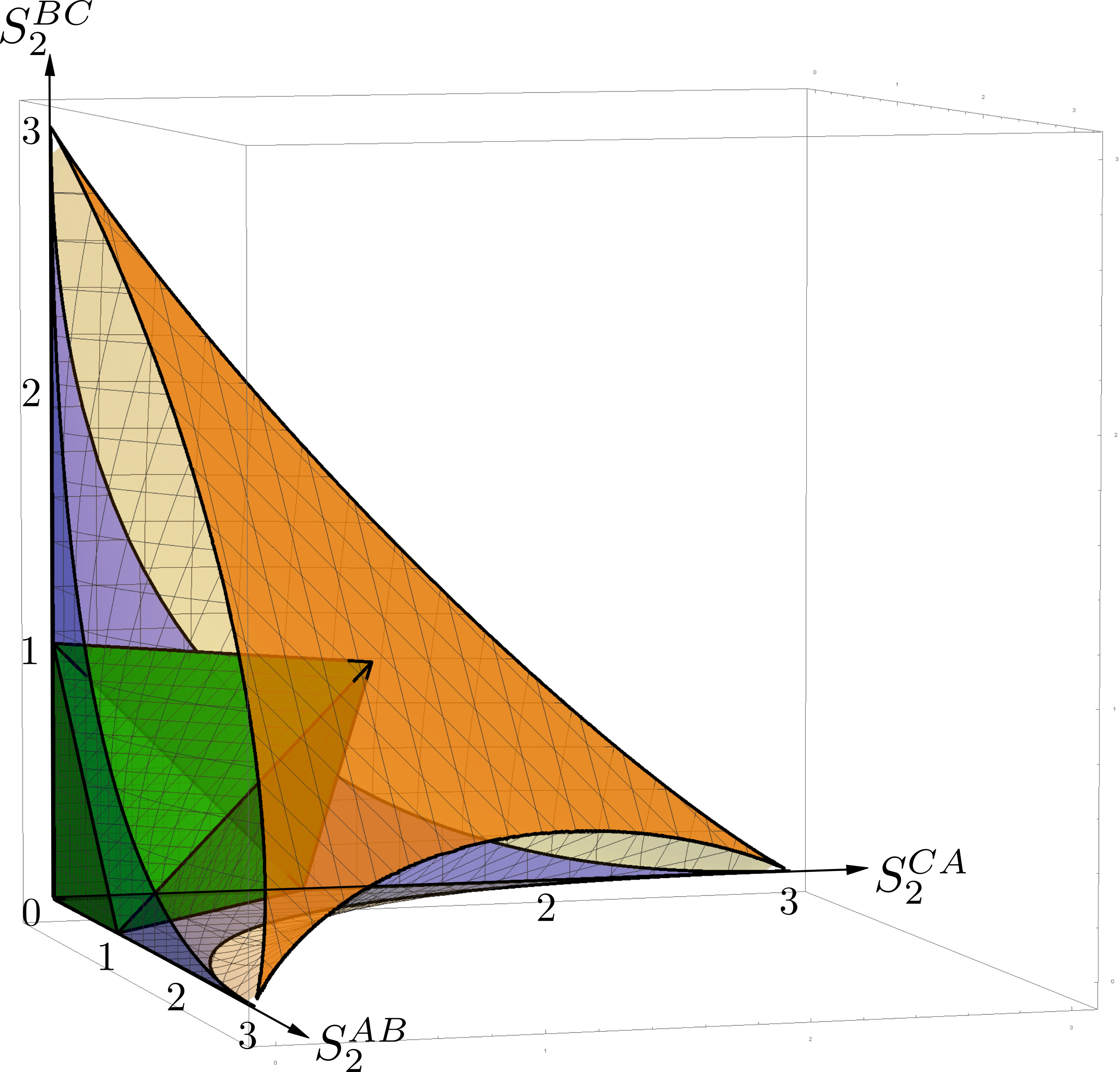}
  \caption{
  The space of the fully separable three-qubit states in the coordinates given by Observation~\ref{ob:Fully_separable_mixed} presented in this paper, along with the pure state in the background. The transparent yellow region specifies the pure state space while the blue is the fully separable region.   }
   \label{fig:Fully_separable_mixed}
\end{figure}

Our criteria only use two-body correlations to refute full separability.
In fact, we will see in Appendix \ref{ap:example_test} that our criterion can be weaker in detecting entanglement than previously known criteria derived in Refs.~\cite{imai2021bound, wyderka2020characterizing}.
So for example, it can detect three-qubit entanglement for noisy asymmetric W states, as shown in Appendix~\ref{ap:example_test}.
But it can never detect entanglement in any permutationally invariant  state, see the previous section, such as mixtures between the symmetric GHZ or W state with white noise.
This may result from the fact that our criterion does not involve three-qubit correlations, {which are included in the other criteria}.
Motivated by this issue, we derive full separability conditions in Eq.~\eqref{fs_stronger} in Appendix \ref{ap:Fully_separable_mixed}, which include 
one- and three-qubit correlations,
which can be stronger than the two previously known criteria, see Appendix \ref{ap:example_test}.

A state is called biseparable for $A|BC$ if it can be written as
\begin{equation}
   \label{bi_sep_form}
    \vr_{A|BC} =
    \sum_k q_k \vr_k^A \otimes \vr_k^{BC},
\end{equation}
with $q_k \in [0,1]$ and $\sum_k q_k =1$.
Now let us propose the following criteria:
\begin{observation}\label{ob:bisep_criteria}
Any three-qubit state which is separable for a bipartition $X|YZ$ obeys
\begin{subequations}
\label{bs_mix}
 \begin{align}
     \begin{split}
         \label{bs_mix_a}
         S^{XY}_2 + S^{ZX}_2 - S^{YZ}_2 &\leq 1,
     \end{split}\\
     \begin{split}
     \label{bs_mix_b}
         S^{XY}_2 , S^{ZX}_2 &\leq 1,
     \end{split} \\
    \begin{split}
     \label{bs_mix_c}
           3 S^{XY}_2 - S^{ZX}_2  + S^{YZ}_2 &\leq 3,
     \end{split} \\
     \begin{split}
       \label{bs_mix_d}
       - S^{XY}_2 + 3 S^{ZX}_2 + S^{YZ}_2 &\leq 3.
     \end{split}
 \end{align}
\end{subequations}
\end{observation}

The proof is given in Appendix~\ref{ap:bisep_criteria} and the sets defined by these inequalities are displayed in Fig.~\ref{fig:bi-separable_mixed}.

For pure biseparable states we recover the previous results, thus states of the form of Eq.~(\ref{eq:gen_bisep}) saturate the inequalities~(\ref{bs_mix_c})-(\ref{bs_mix_d}).
However, it is not clear if the inequalities are also tight for mixed states. Note that if Conjecture~\ref{conj:mixed_total _nonlin} holds, parts of the biseparable sets would lie outside of the region attainable by quantum states. 

A state violating inequality~(\ref{bs_mix_c}) cannot be written in the form of Eq.~(\ref{bi_sep_form}) and therefore is entangled across the bipartition $A|BC$.
Similarly, we can detect entanglement for the partitions $B|CA$ and $C|AB$.
Combining this, any state mapping outside of the union of the three sets defined by Eqs.~(\ref{bs_mix_c})-(\ref{bs_mix_d}) can not be separable with respect to a given partition.
However, this does not imply that the state is genuinely multipartite entangled.
There is the possibility that our inequalities are violated by a mixture of biseparable states.
Such states have the form
\begin{equation}
    \label{eq:gen_bisep}
    \varrho_{\text{bisep}} =
    p_{A} \vr_{A|BC}
    + p_{B} \vr_{B|CA}
    + p_{C} \vr_{C|AB},
\end{equation}
where $p_A, p_B, p_C$ are probability distributions.
In Appendix~\ref{ap:Mixed_bi_sep} we explicitly provide a biseparable state that does not lie in the union of the three different sets corresponding to partition biseparability.

\begin{figure}
  \centering
  \includegraphics[width=0.8\columnwidth]{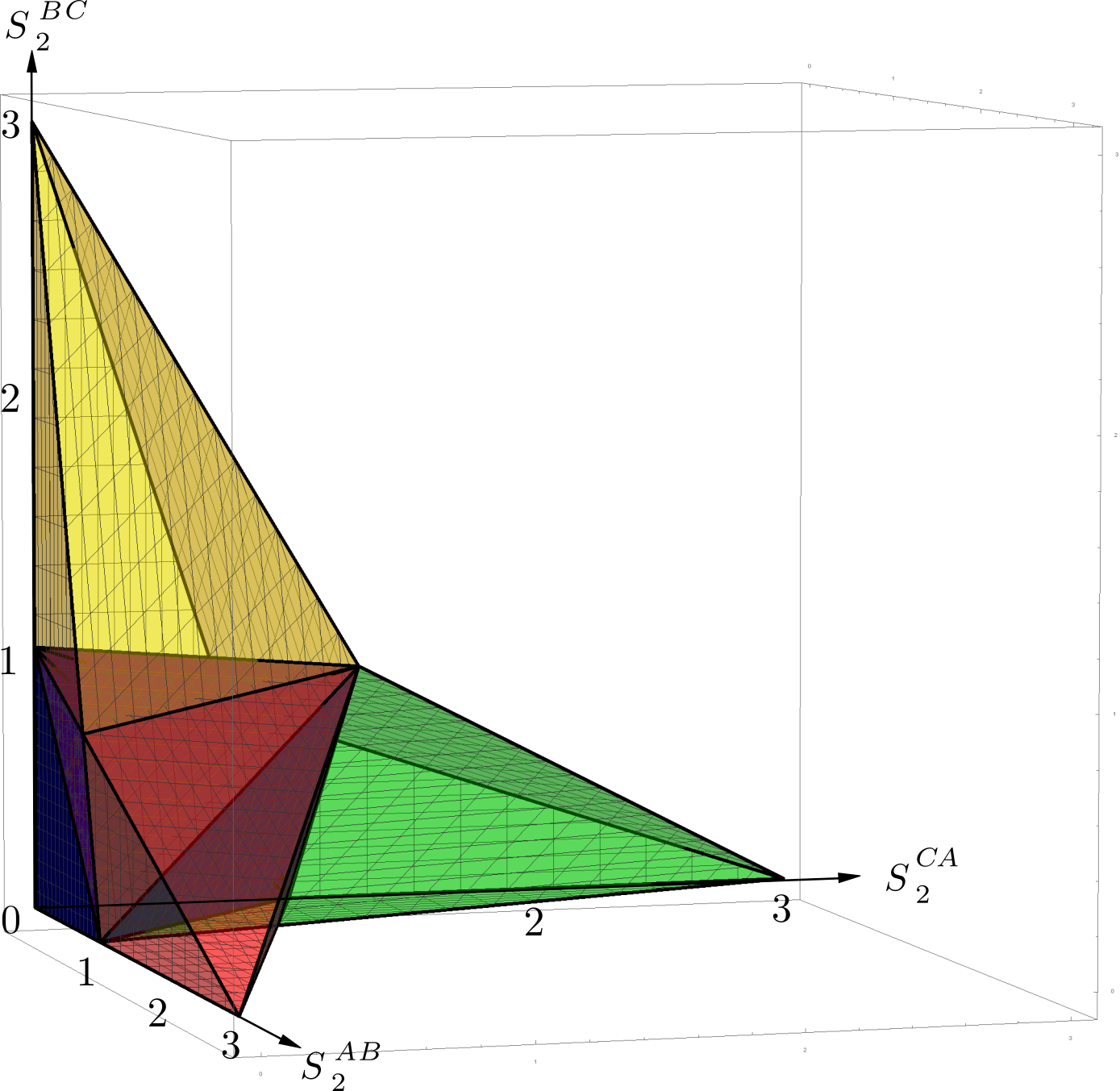}
  \caption{
  The state space of states that are biseparable along some fixed bi-partition, according to Observtion~\ref{ob:bisep_criteria} presented in this paper. The Green, Orange, and Red regions correspond to states biseparable in the $AC|B$, $AB|C$, and $BC|A$ bipartitions.}

   \label{fig:bi-separable_mixed}
\end{figure}

\section{Rank bounds}\label{sec:rank}
In this section, we look at restrictions of the two-body correlations $(S_2^{AB}, S_2^{BC}, S_2^{CA})$ depending on the matrix rank of a quantum state. Since a slight perturbation of a state can arbitrarily increase the rank, it is only possible to find the necessary conditions on the correlations for a state to have a certain rank:
\begin{observation} \label{ob:rank}
    Any rank-two three-qubit state obeys $S_2 \geq 1$ and any rank-three three-qubit state obeys $S_2 \geq 1/3$.
    There is no non-trivial necessary condition for rank-$k$ for $k\geq 4$ three-qubit states in terms of the coordinates $(S_2^{AB}, S_2^{BC}, S_2^{CA})$.
\end{observation}

The proof is given in Appendix~\ref{ap:ob_rank}.
For both rank two and three, there exist states that saturate the inequality. For rank two, one family of states that saturate the bound are the ones in Eq.~\eqref{eq:fs_s} with $p = \frac{1}{2}$.  However, we do not expect these bounds to be optimal, i.e. the tightest bounds in terms of $(S_2^{AB}, S_2^{BC}, S_2^{CA})$, and believe that there exist better non-linear bounds.

Concerning the nonexistence of non-trivial necessary conditions for states of rank four and higher, we note that the trivial one is $S_2 \geq 0$. That is, rank-four states have just a high enough rank to cover the whole set of two-body correlations.
For instance, the separable state
\begin{equation}
    \eta = \frac{1}{4} (
    \ket{000}\!\bra{000} +
    \ket{100}\!\bra{100} +
    \ket{101}\!\bra{101} +
    \ket{110}\!\bra{110})
\end{equation}
has rank $4$ but $S_2 = 0$.
That is, this state is indistinguishable from the maximally mixed state in our coordinates.
The same can even hold for entangled states.
The three-qubit state $\vr_{\theta}$ discussed in Eq.~(F1) in Ref.~\cite{ohst2022certifying}, has no two-body correlation, that is $S_2 = 0$.
This state is separable for all bipartitions, but it is not fully separable. Thus, it contains a very weak form of entanglement in three qubits, called bound entanglement~\cite{horodecki1998mixed}.

\section{Conclusion}
We have developed methods to characterize three-qubit quantum states
from two-body marginal correlations that are invariant under local 
unitaries. For pure states, we have found tight nonlinear bounds and 
given their analytical proofs. For mixed states, we have proposed 
criteria for different forms of multipartite entanglement, which 
are expressed as nontrivial linear combinations.

Our results can be generalized in several directions. First, it would 
be interesting to extend this approach to multi-qubit or high-dimensional 
states. Finding a systematic approach to identify state areas is important
for characterizing the possible marginal configurations.
Second, different types of marginal local unitary invariants may 
be employed as other coordinates. This may lead to new criteria to detect 
multipartite entanglement from marginal information.
Finally, since our entanglement criteria can detect entanglement in which 
the two-body correlations are not equally distributed, our results may 
encourage the characterization of multipartite entanglement without 
symmetric properties.

\section{Acknowledgments}
We would like to thank
Akimasa Miyake,
Ties Ohst,
Jens Siewert,
and
Nikolai Wyderka
for discussions.
We would like to particularly thank Konrad Szyma\'nski for his inputs leading to the analytical proof for Eq.~\eqref{g_bound} in Appendix \ref{ap:bisep_criteria}.
Shravan also thanks Anil Shaji for all the support given during the initial part of the project. 

This work was supported by
the DAAD,
the Deutsche Forschungsgemeinschaft (DFG, German Research Foundation, project numbers 447948357 and 440958198),
the Sino-German Center for Research Promotion (Project M-0294),
the ERC (Consolidator Grant 683107/TempoQ), and 
the German Ministry of Education and Research (Project QuKuK, BMBF Grant No. 16KIS1618K),
the Basque Government through IKUR strategy and through the BERC 2022-2025 program and by the Ministry of Science and Innovation: BCAM Severo Ochoa accreditation CEX2021-001142-S / MICIN / AEI / 10.13039/501100011033, PID2020-112948GB-I00 funded by MCIN/AEI/10.13039/501100011033 and by "ERDF A way of making Europe". 

\onecolumngrid
\newpage
\appendix
\addtocounter{theorem}{-6}
\section{Proof of Observations~\ref{ob:proof_pure} and \ref{ob:proof_threedim}}\label{ap:proof_pure}
\begin{observation}
For any three-qubit pure state the linear condition $S^{AB}_2 + S^{BC}_2 + S^{CA}_2 = 3$ holds. It further holds that 
\begin{align}
    \sqrt{S_2^{AB}} + \sqrt{S_2^{BC}} - \sqrt{S_2^{CA}}
    &\leq \sqrt{3},
\end{align}
and as well for the other two permutations of the party labels.
\end{observation}

\begin{observation}
For any three-qubit pure state the linear condition $S^{AB}_2 + S^{BC}_2 + S^{CA}_2 = 3$ holds. It further holds that
\begin{align}\label{app:eq:pure_3d1}
    S_3^{ABC}&\le 3+ \min(S_2^{AB},S_2^{BC},S_2^{CA}),
\end{align}
and
\begin{align}
    \sqrt{S_2^{AB} + \Delta} + \sqrt{S_2^{BC}+ \Delta}
    - \sqrt{S_2^{CA}+ \Delta}
    &\le \sqrt{3},
\end{align}
where $\Delta \equiv 3 - S_3^{ABC} \in [-1, 2]$ and as well for the other two permutations of the party labels.
\end{observation}

\begin{proof}
First, we note that Observation~\ref{ob:proof_pure} follows from Observation~\ref{ob:proof_threedim}, since this inequality is the least restrictive for $\Delta = 0$.
According to Ref.~\cite{wyderka2020characterizing}, any three-qubit pure state obeys
\begin{align}
    &S_2^{AB} + S_2^{BC} + S_2^{CA} = 3,\\
    &S_1^{A} + S_1^{B} + S_1^{C} + S_3^{ABC} = 4.
\end{align}
The purity of any two marginals is equal, which we reformulate to
\begin{align}
    S_1^{X} + S_1^{Y} + S_2^{XY} &= 1 + 2S_1^{Z},
\end{align}
where $X,Y,Z=A,B,C$.
From this, it follows
\begin{align}\label{eq:ap_A_1}
    S_1^{X} &= 1 + \tfrac{1}{3}S_2^{YZ} - \tfrac{1}{3}S_3^{ABC}.
\end{align}
Eq.~(\ref{app:eq:pure_3d1}) follows immediately by noting that $S_1^{X}\ge0$.
Ref.~\cite{morelli2023correlation} recently showed that for any three-qubit pure states it holds that
\begin{align}\label{eq:result_Jens_Simon}
    \sqrt{S_1^{X}} + \sqrt{S_1^{Y}} - \sqrt{S_1^{Z}} \le 1,
\end{align}
where $X,Y,Z=A,B,C$ and $S_1^X = \sum_{i=1}^3 \ex{\sigma_i}_{\vr_X}^2$ for $X=A,B,C$ and the single-qubit reduced state $\vr_X$.
Substituting Eq.~(\ref{eq:ap_A_1}) into all three permutations of Eq.~(\ref{eq:result_Jens_Simon}) and multiplying with $\sqrt{3}$, we immediately arrive at the inequality in Observation~\ref{ob:proof_threedim}
\begin{align}
    \sqrt{3+S_2^{AB}-S_3^{ABC}} + \sqrt{3+S_2^{BC}-S_3^{ABC}} - \sqrt{3+S_2^{CA}-S_3^{ABC}} &\le \sqrt{3}
\end{align}
and permutations thereof.
Noting that this holds for all values $1\le S_3^{ABC}\le4$ and that the condition is the least restrictive for $S_3^{ABC}=3$, we complete the proof of Observation~\ref{ob:proof_pure}.
\end{proof}

\section{Evidence for Conjecture~\ref{conj:mixed_total _nonlin}}\label{ap:mixed_total _nonlin}
Here we describe the details of getting evidence for Conjecture~\ref{conj:mixed_total _nonlin}:
\begin{conjecture}
For any three-qubit state it holds that
\begin{subequations}
       \begin{align}
        \begin{split}
            \label{Full_ss_ba}
            \sqrt{S_2^{AB}} + \sqrt{S_2^{AC}} - \sqrt{S_2^{BC}}
            &\leq \sqrt{3}, 
        \end{split} \\
        \begin{split}
            \label{Full_ss_ca}
            \sqrt{S_2^{AB}} - \sqrt{S_2^{AC}} + \sqrt{S_2^{BC}}
            &\leq \sqrt{3}, 
        \end{split} \\
        \begin{split}
            \label{Full_ss_da}
            -\sqrt{S_2^{AB}} + \sqrt{S_2^{AC}} + \sqrt{S_2^{BC}}
            &\leq \sqrt{3}.
        \end{split}
       \end{align}
    \end{subequations}
\end{conjecture}

We discuss the case of a rank-two state and the higher ranks as a conjecture.
Let us recall the so-called Acin's decomposition \cite{acin2000generalized}:
any pure three-qubit state can be written as follows, up to local unitaries:
\begin{equation}
    \ket{\psi(\vec \lambda)} = \lambda_0 \ket{000} + \lambda_1 e^{i \phi} \ket{101} + \lambda_2 \ket{110} + \lambda_3 \ket{110} + \lambda_4 \ket{111}, 
\end{equation}
where $ \vec \lambda = (\lambda_0,\dots,\lambda_4)$ and $\sum_{i = 0}^4 \lambda^2_i = 1$.
Thus, up to local unitaries, any three-qubit rank-two state can be written as
\begin{equation}
    \label{general_rank_2}
    \varrho_{2} = p \ketbra{\psi(\vec \lambda)}{\psi(\vec \lambda)} + (1-p) \ketbra{\Psi (\vec a, \vec \phi)}{\Psi(\vec a, \vec \phi)}, 
\end{equation}
where  
\begin{equation*}
    \ket{\Psi(\vec a, \vec \phi )} = a_0 \ket{000} + a_1 e^{i \phi_1} \ket{001} + a_2 e^{i \phi_2} \ket{010} + a_3 e^{i \phi_3} \ket{011} + a_4 e^{i \phi_4} \ket{100} + a_5 e^{i \phi_5} \ket{101} + a_6 e^{i \phi_6} \ket{110}  + a_7 e^{i \phi_7} \ket{111}.
\end{equation*}
Then one can convert the left-hand sides of Eqs.~(\ref{Full_ss_ba}, \ref{Full_ss_ca}, \ref{Full_ss_da})
into the functions with $21$ parameters.
Then one can check that the maximization of the functions can be always bounded by $\sqrt{3}$.

\section{Proof of Observation~\ref{ob:Fully_separable_mixed}}\label{ap:Fully_separable_mixed}
\begin{observation}
Any fully separable three-qubit state obeys
\begin{subequations}
\label{fs_app}
 \begin{align}
     \begin{split}
         \label{fs_app_a}
         S^{AB}_2 + S^{BC}_2 - S^{CA}_2 &\leq 1,
     \end{split} \\
    \begin{split}
    \label{fs_app_b}
           S^{AB}_2 - S^{BC}_2 + S^{CA}_2 &\leq 1,
     \end{split} \\
     \begin{split}
     \label{fs_app_c}
           -  S^{AB}_2 + S^{BC}_2 + S^{CA}_2 &\leq 1.
     \end{split}
 \end{align}
\end{subequations}
These are the optimal criteria in the sense that if there are three $S_2^{XY}$ that obey the inequalities, then there is a fully separable state compatible with them.
\end{observation}

\begin{proof}
Let us begin by recalling the fully separable state
\begin{align}
    \vr_{\text{fs}} =
    \sum_k p_k \vr_k^A \otimes \vr_k^B \otimes \vr_k^C,
\end{align}
with $p_k \in [0,1]$ and $\sum_k p_k =1$.
Here we note that any separable state has positive eigenvalues under partial transposition \cite{peres1996separability}:
\begin{equation}
    \vr_{XY} \in \text{SEP}
\Longrightarrow
\vr_{XY}^{T_X} \geq 0,
\end{equation}
where $(\cdot)^{T_X}$ denotes the partial transposition on subsystem $X$.

In the case where $X=A$ and $Y=BC$, we have that
\begin{align}\label{PPT_1}
    \varrho_1 =
    \sigma_y^{(A)} \vr_{\text{fs}}^{T_A} \sigma_y^{(A)} \geq 0,
\end{align}
where
$\sigma_y^{(A)}$ denotes a Pauli-$Y$ matrix, which is unitary, acting on the subsystem $A$.
Similarly, in the case where $X=BC$ and $Y=A$, we have that
\begin{align}
    \varrho_2\label{PPT_2}
    = \sigma_y^{(B)}\sigma_y^{(C)}
    \vr_{\text{fs}}^{T_{BC}}
    \sigma_y^{(B)}\sigma_y^{(C)}
    \geq 0.
\end{align}
Using the inequality $\tr(AB) \geq 0$ for positive matrices $A,B$, we can arrive at
\begin{align}
    \label{trace_PPT_1}
    & \tr(\vr_{\text{fs}} \varrho_1) \geq 0
    \Longrightarrow
    1 - S_1^{A} + S_1^{B} + S_1^{C} - S^{AB}_2 - S^{CA}_2 + S^{BC}_2 - S_3^{ABC} \geq 0, \\ 
    & \tr(\vr_{\text{fs}} \varrho_2) \geq 0
    \Longrightarrow
    \label{trace_PPT_2}
    1 + S_1^{A} - S_1^{B} - S_1^{C} - S^{AB}_2 - S^{CA}_2 + S^{BC}_2 + S_3^{ABC} \geq 0.
\end{align}

Adding the above two equations, we obtain
\begin{equation*}
    S^{AB}_2 - S^{BC}_2 + S^{CA}_2 \leq 1,
\end{equation*}
which leads us to Eq.~\eqref{fs_app_b}.
Similarly, one can derive Eqs.~\eqref{fs_app_a},\eqref{fs_app_c} by choosing the appropriate bipartitions.
Thus we can complete the proof of Observation~\ref{ob:Fully_separable_mixed}.
\end{proof}

\noindent
\textbf{Remark:}
We can indeed find stronger criteria from Eqs.~(\ref{trace_PPT_1}, \ref{trace_PPT_2}) which, however, require additional knowledge for their evaluation.
For that, we note that
\begin{equation}
    \label{min_prop}
    \text{min}(a,b) = \frac{a + b}{2} - \frac{|a - b|}{2}.
\end{equation}
Then, taking the minimization over both sides leads to the stronger criteria for full separability:
\begin{subequations}
\label{fs_stronger}
 \begin{align}
     \begin{split}
         \label{fs_stronger_a}
         S^{AB}_2 + S^{CA}_2 - S^{BC}_2  + |S_3^{ABC}  + S_1^{A} - S_1^{B} - S_1^{C}| &\leq 1,
     \end{split}\\
     \begin{split}
     \label{fs_stronger_b}
         S^{AB}_2 - S^{CA}_2 + S^{BC}_2  + |S_3^{ABC}  - S_1^{A} + S_1^{B} - S_1^{C}| &\leq 1,
     \end{split} \\
    \begin{split}
     \label{fs_stronger_c}
          -S^{AB}_2 + S^{CA}_2 + S^{BC}_2  + |S_3^{ABC}  - S_1^{A} - S_1^{B} + S_1^{C}| &\leq 1.
     \end{split}
 \end{align}
\end{subequations}
\section{Discussion on three-qubit separability}
\label{ap:example_test}
In this section, we compare various existing separability criteria for three-qubit states with those derived in this work. In the first subsection we compare the full separability criteria and in the second we compare the biseparability criteria. 

\subsection*{D1: Full separability}
 For our first criterion we note that any fully-separable state obeys $S_3^{ABC} \leq 1$\cite{wyderka2020characterizing}. The second criterion is a stronger version of the previous criterion. Any fully separable state satisfies $S_2 + 3 S_3^{ABC} \leq 3 + S_1$ \cite{imai2021bound}.  In the following we will see that the criterion in Eq.~\eqref{fs_app} can be weaker than the above two criteria, while the criterion derived in Eq.~\eqref{fs_stronger} can be stronger than them. 

 To illustrate our assertions, we consider the following five-parameter family of states,
\begin{equation}
\label{five_parameter_family}
    \varrho_1 = g \ketbra{\text{G}(p)}{\text{G}(p)} + w \ketbra{\text{W}(q,r)}{\text{W}(q,r)} + \frac{1 - g - w}{8} \eins,
\end{equation} where  $\left | \text{G}(p) \right \rangle = \sqrt{p} \left|000\right \rangle + \sqrt{1 - p} \left| 111 \right \rangle $ are the generalised GHZ states and $\left | \text{W}(q,r) \right \rangle = \sqrt{q} \left|001\right \rangle +  \sqrt{r} \left|010\right \rangle + \sqrt{1 - q - r} \left| 100 \right \rangle$ are the generalised W states. In table \ref{tab:Full_separability} we compare the detection of full separability for various sub-family of states by four different entanglement detection criteria. 
\begin{center}
\begin{table}[h]
    \centering
   \begin{tabular}{ | P{3.8 cm} | P{2.6 cm} | P{1.6 cm} | P{3 cm} | P{2.2 cm} | P{1.6 cm} | P{2.4 cm} | } 
 \hline  
 \multicolumn{7}{| c |}{Fully separability criteria}\\
 \hline

 Parameter &  States & $S_3 \leq 1 $ & $S_2 + 3S_3 \leq 3 + S_1$  & Eq \eqref{fs} & Eq \eqref{fs_stronger} & Optimal values \\ 
 \hline
 $w = 0, p = \frac{1}{2} $ & {noisy GHZ} & $g \leq 0.5$ & $g \leq \frac{1}{\sqrt{5}} \approx 0.447 $ & All separable & $g \leq \frac{1}{\sqrt 5}$ & $g \leq 0.2$ \\ [0.5 ex]

 $g = 0 , p = \frac{1}{2}, q = r = \frac{1}{3}$ & {noisy W} & $w \leq \sqrt{\frac{3}{11}}$ & $w \leq \frac{3}{\sqrt{41}} \approx 0.468$ & All separable &$w \leq \frac{3}{\sqrt{41}}$ & $w \leq 0.177$ \\ [1.5 ex]

 $g + w= 1 , p = \frac{1}{2}, q = r = \frac{1}{3}$ &  {W-GHZ mixture} & All ent. &  All ent. & All separable & All ent. & -  \\ [0.95 ex]
 
  $g =0 $ & noisy general W   & $w \leq \eta_1$ &   $w \leq \eta_2$ & $w \leq \eta_3$  & $w \leq \eta_4$  & -   \\ [0.5 ex]

  $g =0, q = 0, r = 1/2  $ & noisy Bell state in ab/c bipartition   & $w \leq \frac{1}{\sqrt{3}} $ &   $w \leq \frac{\sqrt 3}{\sqrt{11}}$ \footnotesize{$ \approx 0.522 $} & $w \leq \frac{1}{\sqrt{3}}$  \footnotesize{$ \approx 0.577 $}  & $w \leq \frac{1}{\sqrt 7} $ \footnotesize{$ \approx 0.378 $} &  $w \leq \frac{1}{3} $ \\ [0.5 ex]
 \hline
\end{tabular}
    \caption{Values for which the states are classified as separable for various criteria. The symbols in the fourth row are given in \eqref{table_symbols_fs}. The first three sub-family of states have complete permutational symmetry (i.e., the Bloch representation is invariant under permutation of any two qubits) and hence cannot be detected by our criterion.  }
    \label{tab:Full_separability}
\end{table}
\end{center}

The  symbols $\eta_1 \dots \eta_4$ in the fourth row are given as (note that if the denominator is negative, we ignore it as it would imply that the criterion corresponding to a particular value is always satisfied)
\begin{align}
   \label{table_symbols_fs}
    \eta^2_1 &= \text{min} \left(1,\frac{1}{(1 + 8r - 8(q^2 + r^2 + q(r - 1)))}\right) \\ \notag
    \eta^2_2 &= \text{min} \left(1,\frac{1}{(3 + 32 r - 32 (q^2 + q (-1 + r) + r^2))}\right) \\ \notag
    \eta^2_3 &= \text{min} \left(1,\frac{1}{(1 + 16 q^2 + 16 q (-1 + r) - 8 (-1 + r) r)}, \right. \\ \notag
    &~~~~~~~~~~~~~~~~~~ \left. \frac{1}{(1 + 8 r - 8 ((-1 + q) q + 4 q r + r^2))},\frac{1}{(1 - 8 q^2 + 16 (-1 + r) r + 8 q (1 + 2 r)) }\right) \\ \notag
    \eta^2_4 &= \text{min} \left(1,\frac{1}{(1 + 16 r - 16 (q(q-1) + 2 q r + r^2 ))},
    \frac{1}{(1 - 16r(r-1))}, \frac{1}{(1 - 16q(q-1))}\right) \\ \notag
\end{align}

In the top two plots of Fig. \ref{fig:fs_W_gen} we compare the criterion derived in Eq. \eqref{fs} with $S_3^{ABC} \leq 1$ and  $S_2 + S_3^{ABC} \leq 3 + S_1 $ respectively. Using this, we arrive at the following conclusions. First, the criterion in Eq. \eqref{fs} performs the worst in detecting entanglement for the W state ($p = q= \frac{1}{3}$), and states are close to it. This is expected as the W state is invariant under the permutation of all three qubits. Second, for values of $q,r$ that lie on the edges (i.e. $q = 0, q = r, r = 0$, which correspond to noisy pure-states that are biseparable along a fixed partition.) and do not admit complete permutational symmetry, our criterion seems to perform as well as $S_3^{ABC} \leq 1$, as evidenced by the heat map values going to zero at these values for the respective plot in Fig. \ref{fig:fs_W_gen}. Third, except for the isolated cases of $(p,q) = (0,0),(0,1),(1,0)$(which corresponds to noisy product states), our criterion seems to be universally weaker in detecting entanglement than  $S_2 + S_3^{ABC} \leq 3 + S_1 $.In our last plot in Fig. \ref{fig:fs_W_gen}, we compare the Imai \textit{et al} criterion with the stronger criterion derived in this work in Eq. \eqref{fs_stronger}. The plot shows that our criterion is always better than or equal to the Imai \textit{et al} criterion in detecting entanglement. Furthermore, both these criteria perform equally only on a small subset of states that are close to the symmetric W state.

  \begin{figure}[h]
    \centering
    \includegraphics[width = 0.48 \columnwidth]{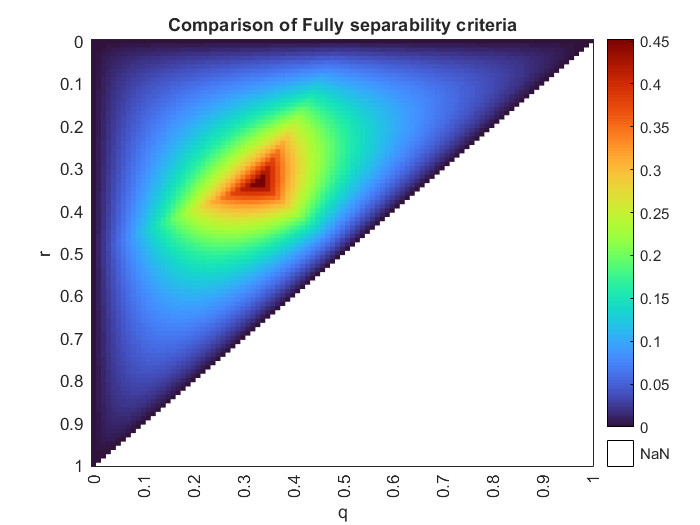}
    \includegraphics[width = 0.48 \columnwidth]{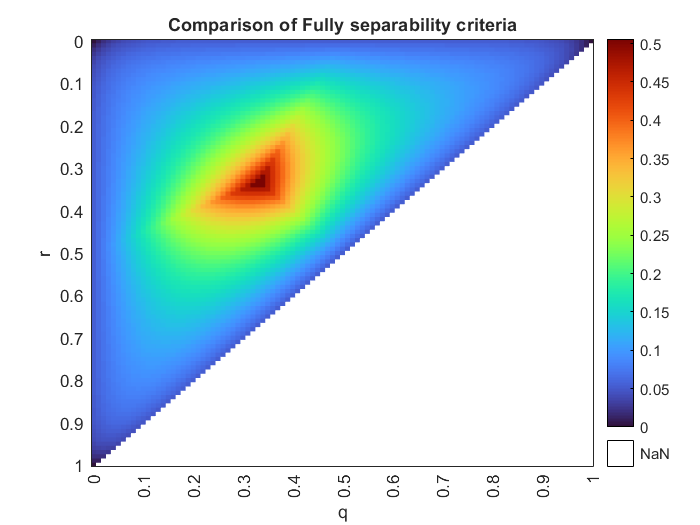}
    \includegraphics[width = 0.48 \columnwidth]{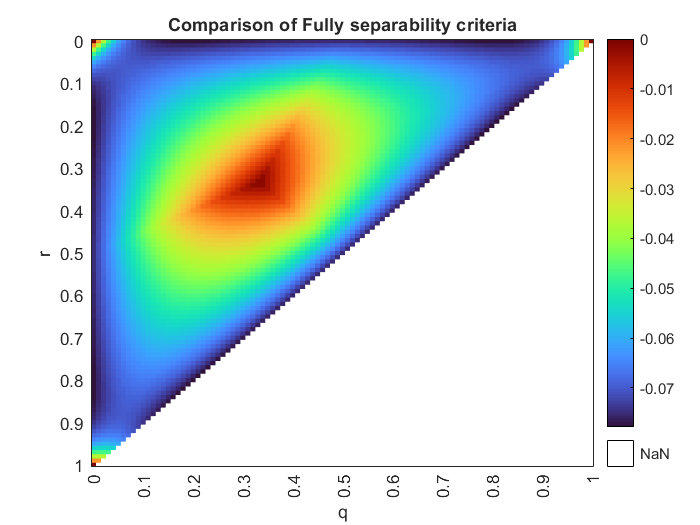}
    \caption{\textbf{Comparison of criteria for full separability}: In the above three plots, we plot the difference between the values of $w$ at which the state is detected as entangled (which we call $w_{ent}$) by the current known criteria and the criterion derived in this work with the parameters of the generalized W state($q,r$). The color bar plots the difference between the $w_{ent}$ estimated by a known criterion and the $w_{ent}$ estimated by our criterion. Thus a positive value implies that the known criterion is better at detecting entanglement, while a negative value implies that our criterion is better at detecting entanglement. NaN refers to the values of $q,r$ are not compatible with $q + r = 1$. In the top-left plot, we compare the criterion $S_3^{ABC} \leq 1$ and Eq. \eqref{fs}. In the top-right plot, we compare the criterion $S_2 + S_3^{ABC} \leq 3 + S_1 $ and Eq. \eqref{fs}. In the bottom plot, we compare the criterion $S_2 + S_3^{ABC} \leq 3 + S_1 $ and Eq. \eqref{fs_stronger}. }
    \label{fig:fs_W_gen}
\end{figure}

\subsection*{D2: Biseparability}

For our first criterion, we note,  all states that are biseparable obey $S_3^{ABC} \leq 3$ \cite{wyderka2020characterizing}. For the second criterion, we note, any state that is biseparable with respect to some partition obeys $S_2 + S_3^{ABC} \leq 3 + 3 S_1$ \cite{imai2021bound}.

We will demonstrate that our biseparability criterion derived in Eq. \eqref{eq_ap:bisep_union} can be more powerful especially when entangled states have less symmetric properties. {Since we do not have a criterion that can detect all biseparable states, we will consider a criterion that can detect most biseperable states. The motivation for choosing such a criterion is to get a qualitative idea of how our biseparability criterion performs with other known criteria and is not meant to be a quantitatively rigorous argument. We consider the criterion given in Eq. \eqref{eq_ap:bisep_union}. Our criterion detects a state as genuinely multipartite entangled if its coordinates lie outside of the union of the three partition biseparable regions. We expect our qualitative conclusions to hold true, since, the mixed biseparable states violations of this criterion (see Appendix \ref{ap:Mixed_bi_sep}) are small.    } 
\begin{equation}
\label{eq_ap:bisep_union}
    \min_{X \neq Y \neq Z \in \{A,B,C\}}\left( 2 | S_2^{XY} - S_2^{YZ} | + S_2 - 3 , S_2^{XY} +  S_2^{YZ} -  S_2^{XZ} - 1 \right) \leq 0. 
\end{equation}

In table \ref{tab:Bi_separability} we compare the detection of biseparability for generalized W states by four different entanglement detection criteria.
\begin{center}
\begin{table}[h]
    \centering
   \begin{tabular}{ | P{3.2 cm} | P{3 cm}  | P{3 cm} | P{2.2 cm} | } 
 \hline  
 \multicolumn{4}{| c |}{Bi -separability criteria}\\
 \hline

 States &  $ S_3 \leq 3$  & $S_2 + 3S_3 \leq 3 + 3 S_1$  & Eq \eqref{eq_ap:bisep_union} \\ 
 \hline

    noisy general W  &  $ w \leq \Xi_1$   &   $ w \leq \Xi_2 $ &   $ w \leq \Xi_3 $  \\ [0.5 ex]

 \hline
\end{tabular}
    \caption{Values for which the states are classified as separable for various criteria. The symbols in the fourth row are given in \eqref{table_symbols_bs}. }
    \label{tab:Bi_separability}
\end{table}
\end{center}

\begin{align}
   \label{table_symbols_bs}
    \Xi^2_1 &= \text{max} \left(1,\frac{3}{(1 + 8r - 8(q^2 + r^2 + q(r - 1)))}\right) \\ \notag
    \Xi^2_2 &= \text{max} \left(1,\frac{3}{\left(32 r - 32 \left(q^2 + r^2 + q(r - 1)\right) -5 \right)}\right) \\ \notag
    \Xi^2_3 &= \text{max} \left(1, \phi_1, \phi_2, \phi_3\right) \\ \notag
&~~~~~~~~~~~ \phi_1 = \text{min} \left(\frac{3}{\left(1 + 8q|2r + q - 1|\right)} , \frac{1}{(1 - 8 q^2 + 16 (-1 + r) r + 8 q (1 + 2 r))}  \right)  \\ \notag
&~~~~~~~~~~~ \phi_2 = \text{min} \left(\frac{3}{\left(1 + 8r|2r + q - 1|\right)} ,\frac{1}{(1 + 16 q^2 + 16 q (-1 + r) - 8 (-1 + r) r)}  \right)  \\ \notag
&~~~~~~~~~~~ \phi_3 = \text{min} \left(\frac{3}{\left(1 + 8|(q - r)(r + q - 1)|\right)} ,\frac{1}{(1 + 8 r - 8 ((-1 + q) q + 4 q r + r^2))}  \right)  \\ \notag
\end{align}

 In Fig. \ref{fig:bs_W_gen} we compare the criterion derived in Eq.\eqref{eq_ap:bisep_union} with $S_3^{ABC} \leq 3$ and  $S_2 +  S_3^{ABC} \leq 3 + 3 S_1 $ respectively. Using this, we arrive at the following conclusions. First, for states that are permutationally invariant and states close to such states, our criterion performs worse than the other two. For non-permutationally invariant states, our criteria seem to be better at detecting entanglement. Second, the edge states, on account of them being noisy biseparable pure states, are detected as separable by the three criteria and hence you see that the performance is the same(heat map is 0). Third, this plot shows that for certain states two-body correlations are more important in determining the bi-partite entanglement than three-body correlations. 

 \begin{figure}[h]
    \centering
    \includegraphics[width = 0.48 \columnwidth]{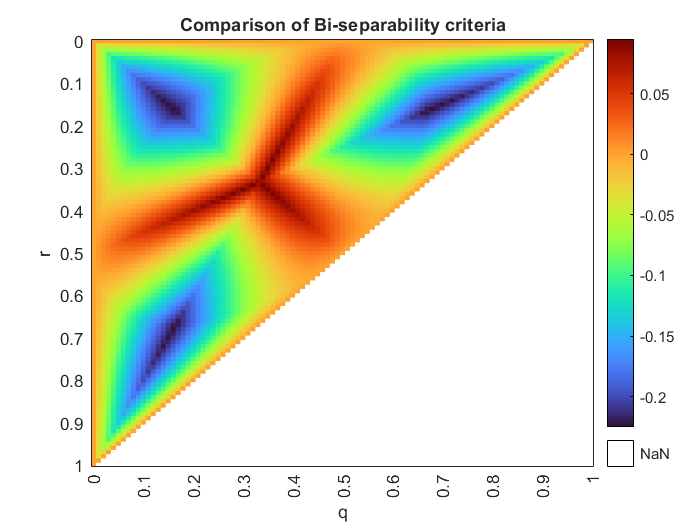} 
    \includegraphics[width = 0.48 \columnwidth]{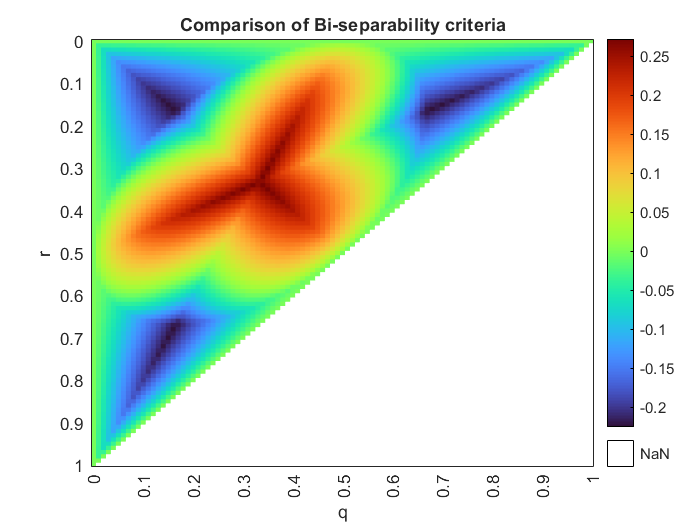}
    \caption{\textbf{Comparison of the biseparble criteria}: In the above three plots, we plot the difference between the values of $w$ at which the state is detected as entangled(which we call $w_{ent}$) by the current known criteria and the criterion derived in this work with the parameters of the generalized W state($q,r$). The color bar plots the difference between the $w_{ent}$ estimated by a known criterion and the $w_{ent}$ estimated by our criterion. Thus a positive value implies that the known criterion is better at detecting entanglement, while a negative value implies that our criterion is better at detecting entanglement. NaN refers to the values of $q,r$ are not compatible with $q + r = 1$. In the left plot, we compare the criterion $S_3^{ABC} \leq 3$ and Eq. \eqref{eq_ap:bisep_union}. In the right plot, we compare the criterion $S_2 + 3 S_3^{ABC} \leq 3 + S_1 $ and Eq. \eqref{eq_ap:bisep_union}.}
    \label{fig:bs_W_gen}
\end{figure}

\section{Proof of Observation~\ref{ob:bisep_criteria}}\label{ap:bisep_criteria}
\begin{observation}
Any three-qubit state which is separable for a bipartition $X|YZ$ obeys
\begin{subequations}
 \begin{align}
     \begin{split}
         S^{XY}_2 + S^{ZX}_2 - S^{YZ}_2 &\leq 1,
     \end{split}\\
     \begin{split}
         S^{XY}_2 , S^{ZX}_2 &\leq 1,
     \end{split} \\
    \begin{split}
           3 S^{XY}_2 - S^{ZX}_2  + S^{YZ}_2 &\leq 3,
     \end{split} \\
     \begin{split}
       - S^{XY}_2 + 3 S^{ZX}_2 + S^{YZ}_2 &\leq 3.
     \end{split}
 \end{align}
\end{subequations}
\end{observation}

\begin{proof}
    Without loss of generality, it is sufficient to prove one of the bipartition.
    In this proof, we take the case where a state is separable
    with respect to partition $A|BC$, which obeys
\begin{subequations}
\label{bs_mix_app}
 \begin{align}
     \begin{split}
         \label{bs_mix_app_a}
         S^{AB}_2 + S^{CA}_2 - S^{BC}_2 &\leq 1,
     \end{split}\\
     \begin{split}
     \label{bs_mix_app_b}
         S^{AB}_2 , S^{CA}_2 &\leq 1,
     \end{split} \\
    \begin{split}
     \label{bs_mix_app_c}
           3 S^{AB}_2 - S^{CA}_2  + S^{BC}_2 &\leq 3,
     \end{split} \\
     \begin{split}
       \label{bs_mix_app_d}
       - S^{AB}_2 + 3 S^{CA}_2 + S^{BC}_2 &\leq 3.
     \end{split}
 \end{align}
\end{subequations}
Note that the inequality (\ref{bs_mix_app_a}) was proven in Appendix~\ref{ap:Fully_separable_mixed} and the inequality \eqref{bs_mix_app_b} is a trivial extension of the pure state criterion, due to the convexity of our coordinates.
In the following, we will then show the inequalities (\ref{bs_mix_app_c}, \ref{bs_mix_app_d}).

We begin by observing the case for pure states.
The biseparability condition for $A|BC$ then leads to the condition $S_2^{AB} = S_2^{CA}$, which we already discussed in the main text.
This implies that the left-hand sides in both Eqs. (\ref{bs_mix_app_c}, \ref{bs_mix_app_d}) become,
    \begin{equation}
       \label{pure_cond_app} 
        3 S^{AB}_2 - S^{CA}_2  + S^{BC}_2 =
        2 S^{AB}_2   + S^{BC}_2   =  S^{AB}_2 + S^{CA}_2  + S^{BC}_2 = 3.
    \end{equation}
That is, the inequalities (\ref{bs_mix_app_c}, \ref{bs_mix_app_d}) are clearly saturated by all rank-$1$ three-qubit states that are biseparable for $A|BC$.
Based on this fact, we will proceed with our proof.

Let us note that any biseparable state for $A|BC$ can be written as
    \begin{equation}
        \varrho_{A|BC} =
        \sum_i p_i
        \ket{\psi^{i}_{A|BC}}\!\bra{\psi^{i}_{A|BC}}.
    \end{equation}    
    Now the squared coefficient of a two-body Bloch component $\sigma^{(m)}_a \otimes \sigma^{(n)}_{b}$ for the qubit indices $m,n =A,B,C$ and Pauli indices $a,b=1,2,3$ are given by
    \begin{align}
        \label{convex_decomp}
        \mean{\sigma^{(m)}_a \otimes \sigma^{(n)}_{b}}_{\varrho_{A|BC}}^2
        = \sum_i p_i^2 \left(\alpha^{m,n}_{i:a,b}\right)^2 + 2 \sum_{i < j } p_i p_j \alpha^{m,n}_{i:a,b} \alpha^{m,n}_{j:a,b},
    \end{align}
where
$\alpha^{m,n}_{i:a,b} = \left\langle \psi^{i}_{A|BC} \right| \sigma^{(m)}_a \otimes \sigma^{(n)}_{b} \left | \psi^{i}_{A|BC} \right \rangle$.
With this expression and the condition in Eq.~(\ref{pure_cond_app}), we can rewrite the left-hand side in Eq.~(\ref{bs_mix_app_c}) as
\begin{equation} \label{lhsbisep}
     3 S^{AB}_2 - S^{CA}_2  + S^{BC}_2
     = 3 \sum_i p^2_i
     +  2\sum_{i < j}p_i p_j  g_{ij}, 
\end{equation}
where the function $g_{ij}$ contains the cross-terms given by
\begin{equation}
 g_{ij}
 = \sum_{a,b =1}^{3} 3\alpha^{A,B}_{i:a,b} \alpha^{A,B}_{j:a,b} -  \alpha^{C,A}_{i:a,b} \alpha^{C,A}_{j:a,b}
 + \alpha^{B,C}_{i:a,b} \alpha^{B,C}_{j:a,b}.
    \label{g_func}
\end{equation}

{Below, we }prove that for any two pure states that are separable in the $A|BC$ bipartition,
it holds that
\begin{equation}
    g_{ij} \leq 3, ~~~ \forall \ket{\psi^{i}_{A|BC}}, \ket{\psi^{j}_{A|BC}} \in \text{SEP}\left(A|BC\right),
    \label{g_bound}
\end{equation}
which results in that Eq.~(\ref{lhsbisep}) can be bounded by $3$.

Let us give this proof in the following by applying a similar technique used to find the maximum value of a separable numerical range, as discussed in Ref.~\cite{szymanski2023numerical}.
First, we rewrite Eq.~\eqref{g_func} as
\begin{equation}
g_{ij} = \mean{M_i}_{\ket{\psi^j_{A|BC}}},
\quad
 M_i \equiv \sum_{a,b =1}^{3} 3\alpha^{A,B}_{i:a,b} \sigma_a^A \sigma_b^B -  \alpha^{C,A}_{i:a,b} \sigma_a^C \sigma_b^A
 + \alpha^{B,C}_{i:a,b} \sigma_a^B \sigma_b^C.
    \label{g_func_1}
\end{equation}
Note that the operator $M_i$, defined in Eq. \eqref{g_func_1} is Hermitian.
Next, since $g_{ij}$ is invariant under local unitaries, without loss of generality, we can fix a biseparable state $\ket{\psi^i_{A|BC}}$ as
\begin{equation}
    \ket{\psi^i_{A|BC}} = \ket{0} \otimes \left(\cos (\theta) \ket{00} +  \sin(\theta)\ket{11}\right).
    \label{ap_bisep}
\end{equation}
Then, the operator $M_i$ depends only on the parameter $\theta$. It is essential to notice that the maximal eigenvalue of $M_i$ is given by $3$, which is independent of $\theta$.
This allows us to show that the maximal value of $g_{ij}$ for $\forall \ket{\psi^{i}_{A|BC}}, \ket{\psi^{j}_{A|BC}} \in \text{SEP}\left(A|BC\right)$ is $3$. 
Hence we can complete the proof of Observation~\ref{ob:bisep_criteria}.
\end{proof}

\noindent
\textbf{Remark:}
The biseparability inequalities can also be written in terms of the purity of total and reduced density matrices of the three-qubit state. This gives us insight into the bounds on the purity of reduced density matrices for a three-qubit biseparable state, similar to the ones that are obtained by Schmidt decomposition for two-qubit states. 
More precisely, the purity of two-qubit and single-qubit reduced density matrix takes the following forms:
\begin{align}
    \label{Marginal_one_trace}
    &\operatorname{tr}(\varrho_{X Y}^2) = \frac{1}{4}\left[ 1 + S_1^{X} + S_1^{Y} + S_2^{XY}\right],
    \\
    &\operatorname{tr}(\varrho_{X}^2) = \frac{1}{2}\left[ 1 + S_1^{X} \right].
    \label{Marginal_two_trace}
\end{align}
Using the Eqs.\eqref{Marginal_one_trace} and \eqref{Marginal_two_trace}, and with a bit of algebra, one can re-write the Eqs. \eqref{bs_mix_app_c} and \eqref{bs_mix_app_d} as,
\begin{subequations}
\begin{align}
    \begin{split}
        \label{purity_f1}
        3 \operatorname{tr}(\varrho_{AB}^2) + \operatorname{tr}(\varrho_{BC}^2)  \leq  \operatorname{tr}(\varrho_{AC}^2)  + \operatorname{tr}(\varrho_{A}^2) + 2\operatorname{tr}(\varrho_{B}^2), 
    \end{split} \\
     \begin{split}
        \label{purity_f2}
        3 \operatorname{tr}(\varrho_{AC}^2) + \operatorname{tr}(\varrho_{BC}^2) \leq \operatorname{tr}(\varrho_{AB}^2)  + \operatorname{tr}(\varrho_{A}^2) + 2\operatorname{tr}(\varrho_{C}^2)  .
    \end{split}
\end{align}
\end{subequations}

\section{Mixed Bi-separable states}
\label{ap:Mixed_bi_sep}
 We note here that the most general biseparable states are classical statistical mixtures of states that are biseparable along different bi-partitions, which do not satisfy the bounds given in Eq. \eqref{bs_mix_app}. One such non-trivial example is given by
 \begin{align}
     \label{eq_app:non_trivial_mix_bs}
     \varrho &= 0.65 \ket{\chi_{A|BC}}\! \bra{\chi_{A|BC}} + 0.35
     \ket{\chi_{C|AB}} \! \bra{\chi_{C|AB}},
 \end{align}
where
\begin{equation}
    \ket{\chi_{A|BC}} = \sqrt{0.97} \ket{000} + \sqrt{0.03} \ket{011},
    \quad
    \ket{\chi_{C|AB}} = - 0.97 \ket{000} - 0.127 \ket{100} + \sqrt{1 - 0.97^2 - 0.127^2} \ket{110},
\end{equation}
The above state is a mixture of states biseparable in the $A|BC$ and $B|CA$ bi-partition. However, for the above state, 
\begin{equation}
    \label{eq_app:non_trivial_viol}
    \min_{X \neq Y \neq Z \in \{A,B,C\}}  2 | S_2^{XY} - S_2^{YZ} | + S_2^{XY} + S_2^{YZ} + S_2^{XZ}  = 3.02214
\end{equation}

We note that the violation is very small and such violations seem to happen only for a small subset of biseparable states. {In fact, we were unable to find numerical values of Eq. \eqref{eq_app:non_trivial_viol} larger than $3.05$. }


\section{Proof of Observation \ref{ob:rank}}\label{ap:ob_rank}
Consider a rank-$n$ density matrix. Suppose its purity is lower bounded by $l$. Then, using $\text{tr}\left(\vr^2\right)  = \frac{1}{8}\left(1 + S_1  + S_2 + S_3^{ABC}\right) $,  we have
\begin{equation}
   \label{eq:rank_bound_1}
     S_1 + S_2 + S_3^{ABC} \geq 8l - 1.
\end{equation}
Furthermore, any three-qubit state obeys
\begin{equation}
 \label{eq:rank_bound_2}
     1 - S_1 + S_2 - S_3^{ABC} \geq 0,
\end{equation}
for details see Ref.~\cite{wyderka2020characterizing}.
Using Eqs.~(\ref{eq:rank_bound_1}, \ref{eq:rank_bound_2}), we arrive at 
\begin{equation}
    S_2 \geq 4l - 1.
\end{equation}
The proof is completed by noting that for rank-two states with $l = \frac{1}{2}$ and for rank-three states with $l = \frac{1}{3}$. 

\bibliographystyle{apsrev4-2}

%

\end{document}